\documentclass[11pt]{amsart}

\usepackage{amsmath,amssymb,amsthm,braket,mathrsfs}

\usepackage{color}
\usepackage{bm}
\newcommand{\bk}{\mathbf{k}}

\newcommand{\bA}{\mathbf{A}}

\newcommand{\be}{\mathbf{e}}
\newcommand{\bp}{\mathbf{p}}

\newcommand{\bx}{\mathbf{x}}
\newcommand{\by}{\mathbf{y}}
\newcommand{\bP}{\mathbf{P}}
\newcommand{\cD}{\mathcal{D}}

\newcommand{\cH}{\mathcal{H}}

\newcommand{\cK}{\mathcal{K}}
\newcommand{\cF}{\mathcal{F}}
\newcommand{\cL}{\mathcal{L}}

\newcommand{\BR}{\mathbb{R}}

\newcommand{\BN}{\mathbb{N}}

\newcommand{\tensor}{\otimes}
\newcommand{\ep}{\epsilon}

\newcommand{\ome}{\omega}

\def\norm#1{\| {#1} \|}

\def\ome{\omega}

\def\ep{\epsilon}
\def\e{{\rm e}}
\def\i{{\rm i}}

\theoremstyle{plain}
\newtheorem{theorem}{Theorem}[section]

\newtheorem{lemma}[theorem]{Lemma}

\theoremstyle{definition}

\theoremstyle{remark}

\title{  Binding condition for a general class of quantum field Hamiltonians }
\author{C. G\'erard}
\address{D\'epartement de Math\'ematiques, Universit\'e de Paris XI, 91405 Orsay Cedex France}
\email{christian.gerard@math.u-psud.fr}
\author{I. Sasaki}
\thanks{I. S.'s work was partly supported by Research supported by KAKENHI Y22740087, 
and was performed through the Program for Dissemination of Tenure-Track System  
funded by the Ministry of Education and Science, Japan}
\address{Fiber-Nanotech Young Researcher Empowerment Center, Shinshu University,
Asahi 3--1--1, Matsumoto 390--8621, Japan.}
\email{isasaki@shinshu-u.ac.jp}

\date{\today}



\begin{document}
%%%%%%%%%%%%%%%%
\maketitle

\begin{abstract}
We consider a system of a quantum particle interacting with a quantum field and an external potential $V(\bx)$.
The Hamiltonian is defined by a quadratic form $H^V = H^0 + V(\bx)$, where 
$H^0$ is a quadratic form which preserves the total momentum. 
$H^0$ and $H^V$ are assumed to be bounded from below. 
 We give a criterion for the positivity of the binding energy $E_\mathrm{bin} = E^0-E^V$,
 where $E^0$ and $E^V$ are the ground state energies of $H^0$ and $H^V$.
As examples of the result, the positivity of the binding energy of the semi-relativistic Pauli-Fierz model and  Nelson type Hamiltonian is proved.
\end{abstract}

%%%%%%%%%%%%%%%%%%%%%%%%%%%%%%%%%%%%%%%%%%%%

%\Markboth{semi-relativistic Nelson}{ground state}

\setlength{\baselineskip}{15pt}

\section{Introduction}
We consider a Hamiltonian of the form
\begin{align}
  H^V = H^0  + V\tensor I, \label{hamil}
\end{align}
acting on the Hilbert space $\cH = L^2(\BR^d;d\bx) \otimes \cK$,
where $\cK$ is a Hilbert space, $H^0$ is a semi-bounded quadratic form on $\cH$
and $V$ is  the operator of multiplication by a real function $V(\bx)$ in $L^2(\BR^d;d\bx)$.
We are interested in the ground state energy $E^V$ of $H^V$.
The {\em binding energy} of the system is defined by
\begin{align}
  E_\mathrm{bin} = E^0 - E^V.
\end{align}
In this paper, we give a criterion for $E_\mathrm{bin}$ to be strictly positive.

Hamiltonians of the form \eqref{hamil} appear in models  of a quantum particle interacting with a quantum field.
One of the important examples is the {\em Pauli-Fierz Hamiltonian}, for which  $d=3$,
$\cK$ is the bosonic Fock space over $L^2(\BR^3\times\{1,2\})$ and
\begin{align}
  H^0 = H_\mathrm{PF}^0 := \frac{1}{2m}(\bp\tensor I+\sqrt{\alpha}\bA(\bx))^2 + I \tensor  H_f
\end{align}
where $H_f$ is the free photon energy, $\alpha$ is the fine structure constant, $\bA(\bx)$ is the quantized vector potential
and $V(\bx)$ is the nuclear potential (see \cite{Griesemer-Lieb-Loss:2001}).
The positivity of the binding energy is used as a hypothesis to establish the existence of  a ground state
of the Pauli-Fierz model in  \cite{Griesemer-Lieb-Loss:2001}.  In  \cite{Griesemer-Lieb-Loss:2001} the 
 positivity of the binding energy is obtained by assuming that 
\begin{align}
  \frac{\bp^2}{2m} +V(\bx)
\end{align}
has a negative energy ground state.
In this paper, we generalize the method developed in \cite{Griesemer-Lieb-Loss:2001} and 
apply it to several types of quantum field Hamiltonians such that the semi-relativistic Pauli-Fierz Hamiltonian,  
 the Pauli-Fierz Hamiltonian with dipole approximation and Nelson type Hamiltonians.

\section{Definitions and Main Results}
If $\cH$ is a Hilbert space we denote by $(\cdot | \cdot)_{\cH}$ the scalar product on $\cH$. If $A$ is a quadratic form on $\cH$, we denote by $Q(A)$ its form domain and the value of $A$ will be denoted by $(\Psi | A \Phi)_{\cH}$ for $\Psi, \Phi\in Q(A)$. We use the same notation for the quadratic form associated to a self-adjoint operator $A$, with domain $Q(A)= {\rm Dom}(|A|^{\frac{1}{2}})$.

We now formulate  the hypotheses of Thm. \ref{theorem} below.

Let $L^2(\BR^d;d\bx)$ be the space  of square integrable functions on $\BR^d$ 
with variable $\bx=(x_1,\dots,x_d)$,
and $\cK$ be a separable complex Hilbert space. We denote by  $\bp=(p_1,\cdots,p_d)=-{\rm i}\nabla_\bx$  the momentum operator on $L^2(\BR^d;d\bx)$ 
The Hilbert space of the total system is:
\begin{align*}
 \cH := L^2(\BR^d;d\bx) \tensor \cK
\end{align*}

We fix a quadratic form 
 $H^0$ on $\cH$ and an 
 external potential $V:\BR^d\to \BR$ which  is a real Borel measurable function. The multiplication by $V(\bx)$ is denoted by the same symbol. 
 
The Hamiltonian of the system is obtained from the  quadratic form on $\cH$ defined by
\begin{align*}
 H^V := H^0 + V.
\end{align*}
We assume the following conditions:
%%%%%%%%%%%%%%%%
\begin{enumerate}
 \item[\bf(H.1)] There exists a dense domain $\cD_0$ such that
\begin{align*}
 \cD_0  \subseteq Q(H^0) \cap Q(V)
\end{align*}
and $H^V$ and $H^0$ are closable and bounded from below on $\cD_0$.
 \item[\bf(H.2)] There exist  a vector of commuting self-adjoint operators $\bP_f=(P_{f,1},\cdots P_{f,d})$ on 
$\cK$ such that $H^0$ commutes with
\begin{align*}
  &  \bP :=  (P_1,\cdots, P_d), \\
  &  P_j = \overline{p_j\tensor I + I \tensor P_{f,j}},
\end{align*}
namely, for all $\bk\in \BR^d$, $\e^{\i\bk\cdot \bP}\cD_0 = \cD_0$ and 
it holds that
\begin{align*}
  (\e^{\i\bk\cdot \bP}\Psi | H^0\e^{\i\bk\cdot \bP}\Phi) &= (\Psi | H^0 \Phi)
\end{align*}
for all $\Psi, \Phi\in \cD_0$ and $\bk \in\BR^d$.
% \item[\bf(H.3)] $V$ strongly commutes with $e^{i \bk \cdot \bx }$ for all $\bk\in\BR^d$.
\end{enumerate}
%%%%%%%%%%%%%%%%
From (H.1), $H^V$ and $H^0$ are closable on $\cD_0$,  and we denote by $\bar{H}^V$, $\bar{H}^0$
the  self-adjoint operators associated to the closure of $H^{V}$, $H^{0}$. Let 
\begin{align*}
 E^V &:= \inf \sigma(\bar{H}^V) 
= \inf_{\Psi \in \cD_0, \norm{\Psi}=1 } (\Psi |H^V\Psi)_{\cH}, \\
 E^0 &:=\inf \sigma(\bar{H}^0) 
 = \inf_{\Psi\in \cD_0, \norm{\Psi}=1} (\Psi | H^0\Psi)_{\cH}.
\end{align*}
be the ground state energies.
The key assumption of the main theorem is the following:
%%%%%%%%%%%%%%%%
\begin{itemize}
 \item[\bf (H.3)] There exist a measurable real function $K(\bk)$ such that 
\begin{align*}
 \frac{1}{2}\{ \Omega(\bk) + \Omega(-\bk) -2 \Omega(0) \} \leq K(\bk) \hbox{ on }\cD_{0}, \ \forall \ \bk\in\BR^d,
\end{align*}
 where $\Omega(\bk):=\e^{-\i\bk\cdot \bx} H^0 \e^{\i\bk\cdot \bx}$.
\end{itemize}
We set 
\begin{align*}
  h := K(\bp) + V
\end{align*}
which is a quadratic form on $L^2(\BR^d;d\bx)$. We assume that
\begin{enumerate}
 \item[\bf(H.4)] There exists a non-trivial subspace $ \cD_1$ of $L^2(\BR^d;d\bx)$ with $\cD_{1} \subset Q(K(\bp)) \cap Q(V)$ such that
for all $f\in \cD_1$ and $\Psi\in \cD_0$,  $f(\bx)\Psi \in \cD_0$. 
Moreover $\cD_1$ is invariant under the complex conjugation, i.e.  $\bar{f}\in \cD_1$ for all $f \in \cD_1$.
\end{enumerate}
We define 
\begin{align*}
  e_0 := \inf_{f\in \cD_1, \norm{f}=1} (f|hf)_{L^{2}}.
\end{align*}
The main theorem in this paper is the following.
%%%%%%%%%%%%%%%%
\begin{theorem}{\label{theorem}}
 Assume the hypotheses (H.1)--(H.4). Then the inequality
\begin{align*}
  E^V \leq E^0 + e_0&
\end{align*}
holds. In particular, if $e_0<0$, then $E_\mathrm{bin} \geq -e_0 > 0$.
\end{theorem}

%%%%%%%%%%%%%%%% begin proof
\section{Proof of Theorem \ref{theorem}}{\label{proof}}
%%%%%%%%%%%%%%%%%%%%%%%%%%%%%%%%%%%%%%%%%%%%%%%%%%
For arbitrary small  $\ep$, we choose normalized vectors $F\in \cD_0$, $f \in \cD_1$ 
such that 
\begin{align*}
  (F|H^0 F)_{\cH} & \leq E^0 + \ep, \\
  (f|h f)_{L^{2}} & \leq e_0 + \ep.
\end{align*}
Since  by (H.4) $h$ commutes with  the  complex conjugation, the function $f$ can be chosen
to be real.
We consider the following extended Hilbert space 
\begin{align*}
  \cH_\mathrm{ex} := L^2(\BR^d; d\by)\tensor \cH,
\end{align*}
which naturally identified with the sets of $\cH$-valued square integrable functions
$L^2((\BR^d;d\by);\cH)$.
For $\by\in\BR^d$, we set $F_\by:=e^{\i\by\cdot \bP}F$ and  consider the $\cH-$valued function:
\begin{align*}
 \Phi : \BR^d\ni \by \mapsto \Phi_\by := f(\bx) F_\by \in \cH .
\end{align*}
The theorem will follow easily from the following three claims:
\begin{equation}
\label{e.0}
\Phi\in \cH_{\rm ex}, \| \Phi\|=1,
\end{equation}
\begin{equation}
\label{e.00}
\Phi \in Q(I\tensor H^{0}), \ (\Phi| I\tensor H^{0}\Phi)_{\cH_{\rm ex}}\leq (F| H^{0}F)_{\cH}+ (f| K(\bp)f)_{L^{2}},
\end{equation}
\begin{equation}
\label{e.000}
\Phi \in Q(I\tensor V), \ (\Phi | I\tensor V\Phi)_{\cH_{\rm ex}}= (f| Vf)_{L^{2}}.
\end{equation}
Let us first prove (\ref{e.0}), (\ref{e.00}) and (\ref{e.000}).  We have:
\begin{align*}
 \int_{\BR^d} \norm{\Phi_\by}_\cH^2 d\by 
&= 
 \int_{\BR^d} \norm{f(\bx)e^{\i\by\cdot \bP}F}_\cH^2 d\by 
  =
\int_{\BR^d} \norm{e^{-\i\by\cdot \bp}f(\bx)e^{\i\by\cdot \bp}F}_\cH^2 d\by   \\
&= \int_{\BR^d} |f(\bx-\by)|^2 d\by  \norm{F}_\cH^2 
= \norm{f}_{L^2(\BR^d)}^2 \cdot  \norm{F}_\cH^2=1,
\end{align*}
which proves (\ref{e.0}). Since $H^{0}$ is bounded below, (\ref{e.00}) will follow from 
\begin{equation}\label{e.1}
(\Phi | I\otimes H^{0}\Phi)_{\cH_\mathrm{ex}}= \int_{\BR^{d}}(\Phi_{\by}|H^{0}\Phi_{\by})_{\cH}d\by\leq (F| H^{0}F)_{\cH}+ (f| K(\bp)f)_{L^{2}},
\end{equation}
using that $F\in Q(H^{0})$ and $f\in Q(K(\bp))$.

Denoting by $\cF: L^2(\BR^d;d\by)\ni f\mapsto\hat{f}\in  L^2(\BR^d;d\bk)$ the unitary Fourier transform, we have:
\[
\begin{array}{rl}
&\int_{\BR^{d}}(\Phi_{\by}| H^{0}\Phi_{\by})_{\cH}d\by=\int_{\BR^{d}}(f(\bx)F_{\by}|H^{0}f(\bx)F_{\by})_{\cH}d\by\\[2mm]
=&\int_{\BR^{d}}(\e^{-\i \by\cdot \bP}f(\bx)\e^{\i \by\cdot \bP}F| \e^{-\i \by\cdot \bP}H^{0}\e^{\i \by\cdot \bP}\e^{-\i \by\cdot \bP}f(\bx)\e^{\i \by\cdot \bP}F)_{\cH}d\by\\[2mm]
=&\int_{\BR^{d}}(f(\bx-\by)F| H^{0}f(\bx-\by)F)_{\cH}d\by\\[2mm]
=&\int_{\BR^{d}}(\e^{\i \bk\cdot \bx}\hat{f}(\bk)F| H^{0}\e^{\i \bk\cdot \bx}\hat{f}(\bk)F)_{\cH}d\bk\\[2mm]
=&\int_{\BR^{d}}|\hat{f}(\bk)|^{2} (F| \Omega(\bk)F)_{\cH}d\bk.
\end{array}
\]
Since $f$ is real valued, we have:
\[
\begin{array}{rl}
&\int_{\BR^{d}}|\hat{f}(\bk)|^{2} (F| \Omega(\bk)F)_{\cH}d\bk\\[2mm]
=&\frac{1}{2}\int_{\BR^{d}}|\hat{f}(\bk)|^{2} (F| (\Omega(\bk)+ \Omega(-\bk)-2 \Omega(0))F)_{\cH}d\bk+ \|f\|^{2}(F|H^{0}F)_{\cH}\\[2mm]
\leq& \|F\|^{2}\int_{\BR^{d}}|\hat{f}(\bk)|^{2} K(\bk)d\bk+ \|f\|^{2}(F| H^{0}F)_{\cH}\\[2mm]
=&(F| H^{0}F)_{\cH}+ (f| K(\bp)f)_{L^{2}},
\end{array}
\]
which proves (\ref{e.00}).

Similarly we have
\[
\begin{array}{rl}
& (\Phi|I\tensor V \Phi)_{\cH_{\mathrm{ex}}}=\int_{\BR^{d}}(f(\bx)F_{\by}| V(\bx)f(\bx)F_{\by})_{\cH}d\by\\[2mm]
=&\int_{\BR^{d}}(\e^{-\i \by\cdot \bP}f(\bx)F_{\by}| \e^{- \i \by\cdot \bP}V(\bx)f(\bx)F_{\by})_{\cH}d\by\\[2mm]
=&\int_{\BR^{d}}(f(\bx-\by)F| V(\bx-\by)f(\bx-\by)F)_{\cH}d\by\\[2mm]
=&(f| V f)_{L^{2}}\|F\|^{2}= (f|Vf)_{L^{2}},
\end{array}
\]
which proves (\ref{e.000}).
From (\ref{e.0}), (\ref{e.00}) and (\ref{e.000}) we obtain
\[
\begin{array}{rl}
E^{V}\leq (\Phi| I\tensor H^{V}\Phi)_{\cH_{\mathrm{ex}}}\leq (F| H^{0}F)_{\cH}+ (f| (K(\bp)+ V)f)_{L^{2}}\leq E^{0}+ e_{0}+ 2\epsilon.
\end{array}
\]
Since $\epsilon$ is arbitrary we obtain the theorem. 
\qed
%%%%%%%%%%%%%%%%

\section{Examples}
%In the following section, we always identify the symmetric operators and its associated quadratic forms.
%%%%%%%%%%%%%%%%
In this section we give some examples to which Thm. \ref{theorem} can be applied.  If $\mathfrak{h}$ is a Hilbert space, we denote by 
\[
\Gamma_{\rm s}(\mathfrak{h})= \bigoplus_{n=0}^\infty\otimes_{\rm s}^{n}\mathfrak{h}
\]
the bosonic Fock space over $\mathfrak{h}$. The vacuum vector  in $\Gamma_{\rm s}(\mathfrak{h})$ will be dentoed by $\Omega$, $a^{*}(h)$, $a(h)$ for $h\in \mathfrak{h}$ denote the  creation/annihilation operators.
\subsection{Semi-relativistic Pauli-Fierz Hamiltonians}
 The {\em semi-relativistic Pauli-Fierz Hamiltonian} is defined as follows: we take $d=3$ and
\begin{align*}
 & \cK = \Gamma_{\rm s}\left( L^2(\BR^3\times \{1,2\})\right), \\
& H^V = H_\mathrm{SRPF}^V := \sqrt{(\bp\tensor I + \sqrt{\alpha}\bA(\bx))^2 + m^2} -m +I\tensor H_f + V\tensor I,
\end{align*}
where $\alpha\in\BR$ is a coupling constant and $m>0$ is the mass of the electron (see \cite{HS10}).
The quantized vector potential $\bA(\bx)$ is defined by
\begin{align}
 \bA(\bx) = \frac{1}{\sqrt{2}} \sum_{\lambda=1,2} \int_{\BR^3} d\bk 
              \frac{\Lambda(\bk)}{|\bk|^{1/2}}\be^{(\lambda)}(\bk)
                (e^{\i\bk\cdot\bx}a_\lambda(\bk) + e^{-\i\bk\cdot\bx}a_\lambda(\bk)), \label{a(x)}
\end{align}
where $a_\lambda^*(\bk), a_\lambda(\bk)$ are creation and annihilation operators on $\cK$,
 $\Lambda$ is a real-function such that $\Lambda, |\bk|^{-1/2}\Lambda \in L^2(\BR^3)$
and  the polarization vectors $\be^{(\lambda)}:\BR^3\to\BR^3$ satisfy
\begin{align*}
  \be^{(\lambda)}(\bk) \cdot \be^{(\lambda')}(\bk) = \delta_{\lambda,\lambda'}, \qquad 
  \be^{(\lambda)}(\bk) \cdot \bk = 0.
\end{align*}
The free photon energy $H_f$ is defined by
\begin{align*}
  H_f = \sum_{\lambda=1,2} \int_{\BR^3}   |\bk| a^*_\lambda(\bk) a_\lambda(\bk)d\bk
\end{align*}
Let
\begin{align*}
 \cF_\mathrm{fin} := \cL[\{a^*(f_1)\cdots a^*(f_n)\Omega_\mathrm{photon}, \Omega_\mathrm{photon}
                                     | f_j \in C_0^\infty(\BR^3\times\{1,2\}), j=1,2,\dots, n, n\in \BN \} ]
\end{align*}
be a finite photon subspace where $\Omega_\mathrm{photon} =(1,0,0,\dots)\in \cK$.
We set 
\begin{align}
& \cD_0 = C_0^\infty(\BR^3) \hat\tensor \cF_\mathrm{fin},\\
& \bP_f = \sum_{\lambda=1,2}\int_{\BR^3} \bk a^*_\lambda(\bk) a_\lambda(\bk) d\bk
\end{align}
where $\hat\tensor $ indicates the algebraic tensor product.
Then the above operator satisfy the condition (H.1) and (H.2).
Moreover, it is proved that (H.3) holds with $K(\bk) = \sqrt{\bk^2+m^2}-m$ (see \cite{HS10}).
We assume that $V\in L_\mathrm{loc}^1(\BR^3;d\bx)$ and set $\cD_1 = C_0^\infty(\BR^3)$.
Then (H.4) holds. Therefore $E_\mathrm{SRPF}^V\leq E_\mathrm{SRPF}^0 + e_0$ holds with
\begin{align*}
& E^\sharp_\mathrm{SRPF} := \inf_{\Psi\in \cD_0, \norm{\Psi}=1} (\Psi|H_\mathrm{SRPF}^V\Psi)_{\cH}, \qquad \sharp=V,0 \\
& e_0 = \inf_{f\in C_0^\infty, \norm{f}=1} (f| (\sqrt{\bp^2+m^2}-m +V)f)_{L^{2}}.
\end{align*}
\subsection{Pauli-Fierz Hamiltonian with dipole approximation}
The Pauli-Fierz Hamiltonian with dipole approximation is defined by
\begin{align}
& H_\mathrm{DP}^V= \frac{1}{2m}(\bp\tensor I+\sqrt{\alpha}\bA(0))^2 + I\tensor H_f + V\tensor I,
\end{align}
where $\bA(0)$ is defined in \eqref{a(x)} with $\bx=0$.
$H_\mathrm{DP}^V$ is defined on $\cD_0=C_0^\infty\hat\tensor \cF_\mathrm{fin}$.
Clearly (H.1) holds.
The operator $H_\mathrm{DP}^0$ is not translation invariant, but it preserves the particle
momentum $\bp$. Hence we set
\begin{align*}
  \bP_f = 0, \quad \bP = \bp.
\end{align*}
Then (H.2) holds. For this Hamiltonian, we have
\begin{align*}
 \frac{1}{2}(\Omega(\bk)+\Omega(-\bk)-2\Omega(0)) = \frac{\bk^2}{2m},
\end{align*}
which implies that (H.3) holds with $K(\bk) = \bk^2/2m$.
We assume that $V\in L_\mathrm{loc}^1(\BR^3)$ and set $\cD_1 = C_0^\infty(\BR^3)$.
Then (H.4) holds. Therefore the inequality $E_\mathrm{DP}^V\leq E_\mathrm{DP}^0 + e_0$ holds with
\begin{align*}
& E^\sharp_\mathrm{DP} := \inf_{\Psi\in \cD_0, \norm{\Psi}=1} (\Psi|H_\mathrm{DP}^V\Psi)_{\cH}, \qquad \sharp=V,0 \\
 &e_0 = \inf_{f\in C_0^\infty, \norm{f}=1} (f|(\frac{\bp^2}{2m} +V)f)_{L^{2}}.
\end{align*}

\subsection{Nelson type Hamiltonians}
We define the Nelson type Hamiltonian by
\begin{align*}
 & \cK =\Gamma_{\rm s} (L^2(\BR^d)), \\
&  H^V = H_\mathrm{Nel}^V := B(\bp^2)\tensor I + I\tensor H_f + P(\phi(\bx)),
\end{align*}
where $B:\BR_+\to \BR_+$ is a Bernstein function, i.e.,
\begin{align}
  B(u) \geq 0, \quad B(0)=0, \quad   (-1)^n \frac{d^nB(u)}{du^n}  \geq 0, \quad n=1,2,\dots. \label{bern}
\end{align}
The field operator $\phi(\bx)$ is defined by
\begin{align*}
 \phi(\bx) = \frac{1}{\sqrt{2}}\int_{\BR^d}  (g(\bk)e^{-\i\bk\cdot\bx}a^*(\bk) + \overline{g(\bk)}e^{\i\bk\cdot\bx}a(\bk))d\bk
\end{align*}
with $g \in L^2(\BR^d)$,  $a^*, a$ are creation and annihilation operators on $\cK$ and $P$ is a real, bounded below  polynomial.

The free boson Hamiltonian $H_f$ is defined by
\begin{align*}
  H_f = \int_{\BR^d} \ome(\bk) a^*(\bk) a(\bk) d\bk,
\end{align*}
where $\ome$ is a non-negative function.
We refer the reader to \cite{LHB11} for a recent study of the Nelson-type Hamiltonians with Bernstein function type kinetic energy . We set 
\begin{align*}
& \cF_\mathrm{fin} := \cL[\{a^*(f_1)\cdots a^*(f_n)\Omega_\mathrm{b}, \Omega_\mathrm{b}
                                     | f_j \in C_0^\infty(\BR^3), j=1,2,\dots, n, n\in \BN \} ],\\
&  \cD_0 = C_0^\infty(\BR^d)\hat\tensor \cF_\mathrm{fin}
\end{align*}
where $\Omega_\mathrm{b}=(1,0,0,\dots)\in\cK$. Assume that $V \in L_\mathrm{loc}^1(\BR^d)$.
By \eqref{bern}, we have
\begin{align*}
  B(u) \leq \frac{u^3}{6} + B''(0)\frac{u^2}{2} + B'(0)u.
\end{align*}
Hence, $C_0^\infty(\BR^d)\subset {\rm Dom}(B(\bp^2))$.
Then $H_\mathrm{Nel}^V$ and $H_\mathrm{Nel}^0$ are well-defined on $\cD_0$ and (H.1) holds.
We set 
\begin{align*}
 \bP_f = \int_{\BR^d} \bk a^*(\bk)a(\bk) d\bk.
\end{align*}
Then, $H_\mathrm{Nel}^0$ commutes with $P_j=\overline{p_j\tensor I+I\tensor P_{f,j}}, j=1,\dots,d$ and (H.2) holds.
Next we check (H.3). We note that 
\begin{align*}
  \Omega(\bk) + \Omega(-\bk) -2 \Omega(0) = B((\bp+\bk)^2) + B((\bp-\bk)^2) -2B(\bp^2).
\end{align*}
We have the following lemma:
%%%%%%%%%%%%%%%%
 \begin{lemma}{\label{lem}}
For all $\bp,\bk\in\BR^d$, the inequality
\begin{align*}
 \frac{1}{2}\left(B((\bp+\bk)^2) + B((\bp-\bk)^2) -2B(\bp^2)  \right)  \leq B(\bk^2).
\end{align*}
holds.
 \end{lemma}
%%%%%%%%%%%%%%%%
 \begin{proof}
It is known that any Bernstein function can be written in the form
\begin{align}
  B(u) = a + bu + \int_{\BR_+} (1-e^{-tu})\mu(dt), \quad (u\geq 0) \label{31}
\end{align}
where $a,b$ are  non-negative constants and $\mu$ is a non-negative measure on $\BR_+$ such that 
$\int_{\BR_+} \min\{t,1\} \mu(dt)<\infty$ (see \cite{LHB11}).
 Hence it is sufficient to prove the inequality
\begin{align}
  -e^{-(\bp+\bk)^2t}  - e^{-(\bp-\bk)^2t} + 2 e^{-\bp^2t} \leq 2(1-e^{-\bk^2t}),  \label{32}
\end{align}
for all $t\geq 0$ and  $p,k\in\BR^d$. If $t=0$, \eqref{32} is trivial.
Without loss of generality, one can set $t=1$. 
Moreover we can assume that $k=(\kappa,0,0)$, $\kappa\geq 0$ by the spherical symmetry of \eqref{32}. 
Then \eqref{32} will  follow from 
\begin{align*}
  b_{\kappa}(p_1):= -e^{-(p_1+\kappa)^2} - e^{-(p_1-\kappa)^2} + 2 e^{-{p_1}^2} \leq 2(1-e^{-\kappa^2}),
\end{align*}
where $\bp=(p_1,p_2,p_3)$. It is enough to show that $b_\kappa(p_1) \leq 2(1-e^{-\kappa^2})$ for $\kappa> 0$ and $p_1>0$.
We set $p_1=a\kappa$ with $a> 0$. Then
\begin{align*}
  b_\kappa(p_1)  &= e^{-a^2\kappa^2}\left[-e^{-\kappa^2}(e^{-2a\kappa^2}+e^{2a\kappa^2})+2\right] \\
          & \leq e^{-a^2\kappa^2}\left[-2e^{-\kappa^2} + 2\right] \\
          & \leq 2(1-e^{-\kappa^2}),
\end{align*}
where we used the inequality $e^{-2a\kappa^2}+e^{2a\kappa^2}\geq 2$ and $e^{-a^2\kappa^2}\leq 1$.
 \end{proof}
Lemma \ref{lem} implies that (H.3) holds with $K(\bk) = B(\bk^2)$.
By setting $\cD_1 = C_0^\infty(\BR^d)$, (H.4) holds. Therefore, by Theorem \ref{theorem}, $E^V_\mathrm{Nel} \leq E^0_\mathrm{Nel} +e_0$ holds with
\begin{align*}
& E^\sharp_\mathrm{Nel} := \inf_{\Psi\in \cD_0, \norm{\Psi}=1} (\Psi|H_\mathrm{Nel}^V\Psi)_{\cH}, \qquad \sharp=V,0 \\
 & e_0 := \inf_{f\in \cD_1, \norm{f}=1} (f|(B(\bp^2)+V)f)_{L^{2}}.
\end{align*}

%\bibliographystyle{amsplain}
%\bibliography{ineq}

\begin{thebibliography}{1}

\bibitem{HS10}
F.~Hiroshima and I.~Sasaki, \emph{On the ionization energy of semi-relativistic
  Pauli-Fierz model for a single particle}, Kokyuroku Bessatsu \textbf{B21}
  (2010), 25--34.

\bibitem{LHB11}
J.~L{\H{o}}rinczi, F.~Hiroshima, and V.~Betz, \emph{{F}eynman-{K}ac-type
  theorems and {G}ibbs measures on path space}, vol.~34, Walter De Gruyter,
  2011, Seminar on Probability, Studies in Mathematics.

\bibitem{Griesemer-Lieb-Loss:2001}
E.~H.~Lieb M.~Griesemer and M.~Loss, \emph{Ground states in non-relativistic
  quantum electrodynamics}, Invent Math \textbf{145} (2001), no.~1, 557--595.

\end{thebibliography}
\providecommand{\bysame}{\leavevmode\hbox to3em{\hrulefill}\thinspace}
\providecommand{\MR}{\relax\ifhmode\unskip\space\fi MR }
% \MRhref is called by the amsart/book/proc definition of \MR.
\providecommand{\MRhref}[2]{%
  \href{http://www.ams.org/mathscinet-getitem?mr=#1}{#2}
}
\providecommand{\href}[2]{#2}

\end{document}